\documentclass{article}

\usepackage{amsmath,amssymb,amsthm,fullpage,enumerate,graphicx,cite}

\newtheorem{lemma}{Lemma}
\newtheorem{corollary}{Corollary}
\newtheorem{claim}{Claim}
\newtheorem{theorem}{Theorem}
\newtheorem{observation}{Observation}

\title{Covering nearly surface-embedded graphs with a fixed number of balls}

\author{Glencora Borradaile  \and
  Erin Wolf Chambers}

\begin{document}
\maketitle

\begin{abstract}
  A recent result of Chepoi, Estellon and Vaxes [Disc.\ Comp.\ Geom.\
  '07] states that any planar graph of diameter at most $2R$ can be
  covered by a constant number of balls of size $R$; put another way,
  there are a constant-sized subset of vertices within which every other
  vertex is distance half the diameter.  We generalize this
  result to graphs embedded on surfaces of fixed genus with a fixed
  number of apices, making progress toward the conjecture that graphs
  excluding a fixed minor can also be covered by a constant number of
  balls.  To do so, we develop two tools which may be of independent interest. The first gives a bound
  on the density of graphs drawn on a surface of genus $g$ having a
  limit on the number of pairwise-crossing edges.  The second bounds
  the size of a non-contractible cycle in terms of the Euclidean norm
  of the degree sequence of a graph embedded on surface.

  % \keywords{topological graphs \and graph covering \and crossing
  %   number \and separators}
  % %\subclass{MSC code1 \and MSC code2 \and more}  
\end{abstract}

\section{Introduction}

Chepoi, Estellon and Vax\`{e}s showed there is a constant $\rho$ such
that any planar graph of diameter at most $2R$ has a subset of at most
$\rho$ vertices such that every vertex in the graph is within distance
$R$ of that subset~\cite{planarballcover}. Since this can be viewed as
showing that there is a constant-sized set cover in the set system of
balls of radius $R$, we refer to this property as the {\em ball-cover}
property.  Graphs having constant-sized ball covers admit interval
routing schemes with dilation $\frac{3}{2} \times \text{diameter}$ and
compactness $O(1)$ where dilation measures the indirectness of the
routing scheme and compactness measures the size of the routing
table~\cite{GPRS01}.  We believe the ball-cover property is an
inherently interesting property.  Graphs having this property could
define an interesting class of graphs and perhaps could have broader
utility than previously realized.

We generalize the class of graphs having the ball-cover property to
those graph families that can be embedded on a surface of fixed genus
after the removal of a constant number of vertices (the {\em apices});
the number of balls required depends only on the genus of the surface
(either orientable or non-orientable) and the number of apices.  Since
graphs of bounded treewidth are also known to have the ball-cover
property~\cite{GPRS01} by way of the Graph Minor Structure Theorem,
our result is a significant step toward proving that
fixed-minor-excluded graphs also have the ball-cover property.  
We discuss this more in
Section~\ref{sec:apex}.  We start by sketching the proof for the
planar case as we use a similar, but more general, tool set here.

\subsection{A sketch of the proof of the ball-cover property for planar graphs}

Throughout, graphs are simple, undirected and unweighted.  Let $B(x)$
be the set of all vertices that are within distance $R$ of vertex $x$
in graph $G$; this is the {\em ball centered at $x$}. Let ${\cal B}(G)
= \{B(x)\ : \ x \in V(G)\}$; this is the {\em ball system} of $G$. We
say that ${\cal B}' \subset {\cal B}$ covers $G$ if ${\cal B}'$ is a
set cover of $V(G)$.

The dual of a set system $\cal S$ with ground set $U$ is defined as
follows: the ground set of the dual set system is $\cal S$ and for
every element $x \in U$, the dual system has a set representing the
sets of $\cal S$ containing $x$, i.e., $X = \{ S\ : \ S \in {\cal S},
x \in S\}$.  It is easy to see:
\begin{observation} \label{obs:dual}
  The dual set system of ${\cal B}(G)$ is ${\cal B}(G)$.
\end{observation}
Since a hitting set of $\cal S$ (a subset of the ground set that
contains an element in every set) is a set cover of the dual set
system of $\cal S$, we likewise have that the centers of a subset of
balls covering $G$ is a hitting set for the ball system.  A hitting
set of ${\cal B}(G)$ is exactly a subset of vertices within which
every other vertex is distance $R$.

Matou\u{s}ek gives a characterization of set systems that have small
hitting sets~\cite{Matousek04} in terms of the set system's {\em
  fractional-Helly} or {\em $(p,q)$-property} and the dual set
system's {\em VC-dimension}.

\subsubsection*{VC-dimension} A set system $\cal S$ {\em
  shatters} a set $X$ if for every subset $Y$ of $X$ there is a set $S
\in {\cal S}$ such that $S \cap X = Y$. The Vapnik-Chervonenkis
dimension or VC-dimension of $\cal S$ is the maximum size of a set
that $\cal S$ can shatter~\cite{VC71}.  Chepoi et~al. remark that
the VC-dimension of the ball system of a graph excluding $K_{r+1}$ as
a minor is at most $r$~\cite{planarballcover}.  This gives us:
\begin{lemma}\label{lem:vc-dim}
The VC-dimension of ball system of a graph excluding $H$ as a
minor is at most $|H|-1$.
\end{lemma}
Recall that a minor of a graph $G$ is a graph that is obtained from
$G$ by edge contractions and deletions; a forbidden or excluded minor
is a graph that {\em cannot} be obtained this way.  It follows from
Observation~\ref{obs:dual} that the dual of the ball system of a graph
excluding $K_{r+1}$ as a minor also has VC-dimension at most $r$.

\subsubsection*{Fractional Helly theorems} If a set system is such that
every $d$ sets has a point in common, then the set system is said to
have Helly order $d$. A {\em Helly theorem} is one that shows that
certain set systems of Helly order $d$ have a non-empty intersection.
The first such theorem was given for the Euclidean plane: if a family
of convex sets has a nonempty intersection for every triple of sets,
then the whole family has a nonempty intersection~\cite{Helly23}.  A
set system has {\em fractional} Helly order $(p,q)$, or {\em has the
  $(p,q)$-property}, if among every $p$ sets some $q$ have a point in
common. Matou\u{s}ek gave the following fractional Helly theorem:
\begin{theorem}[Fractional Helly Theorem~\cite{Matousek04}]\label{thm:pq}
  Let $\cal Q$ be a set system having the $(p,q)$-property (for
  $p \ge q$) and whose dual set system has VC-dimension $q-1$.  Then there
  is a constant $\rho$ such that $\cal Q$ has a hitting set of size at
  most $\rho$.
\end{theorem}

Given Lemma~\ref{lem:vc-dim}, one could therefore show that, for a
fixed minor $H$, $H$-minor free graphs have the ball-cover property
by showing that the corresponding ball system has fractional Helly order $(p,
|H|)$ for some fixed $p \ge |H|$.  Chepoi et~al.\
do just this for planar graphs.  Starting with $p$ vertices, they consider the pairwise
shortest paths between these vertices; each shortest path contains a
vertex that is contained by the balls centered on the paths'
endpoints.  Viewing these shortest paths as edges of a complete graph
and drawn on the plane (as inherited from a drawing of the original
graph), they invoke a result showing that such a drawing of $K_p$, for $p$ sufficiently large, must contain at least 7 pairwise
crossing edges.  The 7 pairwise crossing shortest paths then witness a
point in common to 5 of the balls.  We use this idea at the heart of
our proof for surface-embedded graphs.

\subsection{Surface-embedded graphs}

We start by extending this result to graphs embedded on more general
surfaces. We first give some definitions.

 A $2$-manifold (or surface) $S$ is a Hausdorff space in
which every point has a neighborhood homeomorphic to the Euclidean
plane or the closed half plane.  A cycle in a surface is a continuous
function from $S^1$ to the surface; the cycle is called simple if the
map is injective. A simple cycle $\gamma$ is separating if $S
\backslash \gamma$ is not connected; see Figure~\ref{fig:homologous}.  The genus $g$ of a surface $S$
is the maximum number of pairwise disjoint non-separating cycles $\gamma_1,
\gamma_2, \ldots, \gamma_g$ such that $S \setminus (\gamma_1 \cup \cdots \cup \gamma_g)$ is
connected.  
Note that cutting a surface along a non-separating cycle reduces the genus by 1; this is a common
algorithmic technique for reducing the complexity of a surface.
A surface is non-orientable if and only if it contains a
subspace homeomorphic to the M\"{o}bius band and is 
otherwise orientable.

An embedding of a graph $G=(V,E)$ on a surface $S$ is a drawing of $G$
on $S$, such that vertices are mapped to distinct points in $S$ and
edges are mapped to \emph{internally} disjoint simple paths.  A face
of an embedding is a maximal connected subset of $S$ that does not
intersect the image of $G$.  An embedding is cellular if all of its
faces are homeomorphic to a topological open disc.  We say that $G$ is
a graph of (orientable or non-orientable) genus $g$ if $G$ has
a cellular embedding on a surface of (orientable or non-orientable)
genus $g$.

We will briefly use the notion of $\mathbb{Z}_2$-homology in this paper and so include a 
brief description for completeness; we refer the reader to a topology text for full 
details~\cite{h-at-02,m-t-00}.  A homology cycle is a 
linear combination of oriented cycles with coefficients from a ring
$R$; when $R=\mathbb{Z}_2$, these homology cycles are even-degree subgraphs of $G$.  
A boundary subgraph is the boundary of a union of faces of $G$.  
Two subgraphs are homologous if their symmetric difference is a boundary 
subgraph, or, more intuitively, if
they can be deformed to each other (where the deformation may include
merging intersection cycles or splitting at self-intersections or
deleting trivial separating cycles); see Figure~\ref{fig:homologous} for an example.
Boundary cycles are null-homologous, and since every separating cycle is a boundary cycle, we can view separating cycles as the identity element for homology classes.

\begin{figure}
\begin{center}
\includegraphics[width=3in]{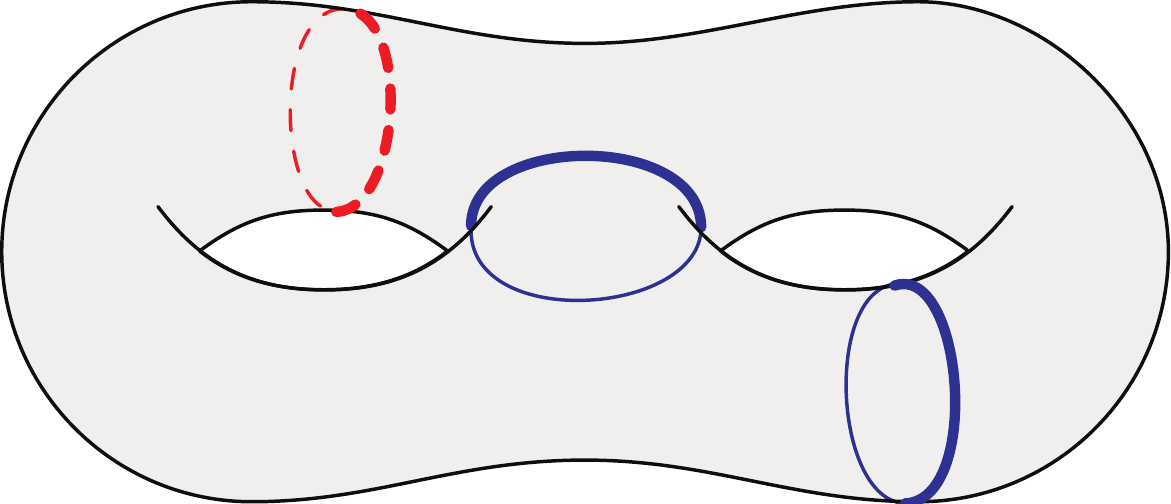}
\includegraphics[width=3in]{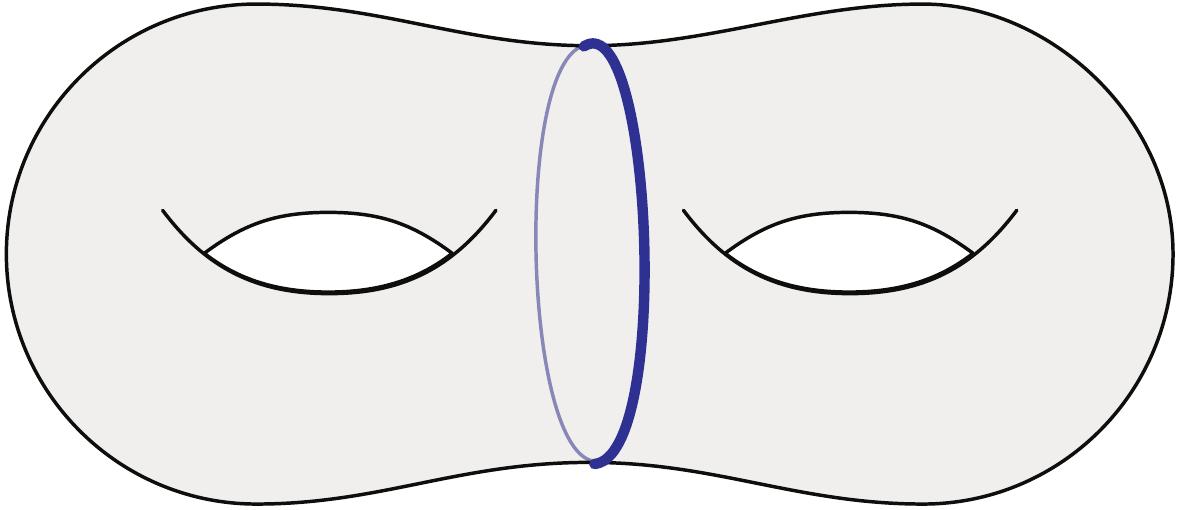}
\end{center}
\caption{Left: An example of homologous cycles: the single dashed red non-separating cycle (above on left) is $\mathbb{Z}_2$-homologous to two solid blue cycles. Right: A null-homologous separating cycle.}
\label{fig:homologous}
\end{figure}

\subsection{Our contribution}

The bulk of this paper focusses on showing that a graph of genus $g$
has the ball-cover property by showing that its ball system has the
$(p_g,q_g)$-property for numbers $p_g$ and $q_g$ that depend only on
$g$ (Section~\ref{sec:pq}). Since $K_n$ has orientable genus $\lceil \frac{1}{12}(n-3)(n-4)
\rceil$ and non-orientable genus $\lceil \frac{1}{6}(n-3)(n-4)
\rceil$\cite{ringel-youngs68}, we set $q_g = c\cdot g^2$ (where $c$
depends only on whether the surface in question is orientable).  Then, since a
graph of genus at most $g$ excludes $K_{q_g}$ as a minor, the
VC-dimension for a graph of genus at most $g$ is at most $q_g-1$.  By
Observation~\ref{obs:dual}, Lemma~\ref{lem:vc-dim} and the Fractional
Helly Theorem, we will get:
\begin{theorem} \label{thm:main}
  There exists a constant $\rho_g$ (depending only on $g$) such that
  any graph of genus at most $g$ and diameter at most $2R$ can be
  covered by at most $\rho_g$ balls of radius $R$.
\end{theorem}
We show that the same holds if the graph additionally has a fixed
number of apices and discuss how one might generalize to
fixed-minor-excluded graph families in Section~\ref{sec:apex}.

In order to prove that the ball system for a genus-$g$ graph has the
$(p_g,q_g)$-property, we show that there is a small set of edges of a
surface-embedded graph whose removal leaves a planar graph
(Section~\ref{sec:sep}) and give bounds on the number of edges in a
graph drawn on a surface of fixed genus having a limit on the number
of crossings (Section~\ref{sec:cross}).  The former result can be used
to generalize an edge-separator result for planar graphs due to Gazit
and Miller~\cite{GaMi90}.  Both these
results are likely of more general interest.  We give background on
these problems in their relevant sections.  

The takeaway from these generalizations will allow us to argue that
any topological drawing of $K_n$ on a surface of orientable or
non-orientable genus $g$ must have a large subset of edges that
pairwise cross.
% \begin{theorem}\label{thm:pwcross}
%   TODO FIX THIS
%   Every topological drawing of $K_n$ on an orientable or
%   non-orientable surface of genus $g$ has a subset of at least
%   $\kappa \log n/\log(g\log n)$ edges that pairwise cross where
%   $\kappa$ is a fixed constant.
% \end{theorem}
In Section~\ref{sec:cross}, we will formally define what constitutes a
topological drawing on a surface of genus $g$ and prove this theorem.

\section{A norm-sized, planarizing edge set for surface-embedded graphs} \label{sec:sep}

In this section, we show there is a small set of edges in a surface-embedded graph whose removal leaves a planar graph.  We start by bounding the size of a non-separating cycle:
\begin{theorem}\label{thm:cycle}
  The shortest non-separating cycle of a graph $G$ embedded on a surface has length at most ${1\over 2}||G||_f$.
\end{theorem}
where 
\[
||G||_f = \sqrt{\sum_{f \in {\cal F}} |f|^2}
\]
is the {\em face-norm} of $G$ and ${\cal F}$ is the set of faces of
$G$.  We use a sequence of $g$ non-separating cycles to {\em
  planarize} $G$.  The face-norm was used by Gazit and Miller to
tighten the bound on the size of edge-separators for planar
graphs~\cite{GaMi90}.  Theorem~\ref{thm:cycle} implies an $O(g
||G||_f)$-sized edge separator for genus-$g$ graphs.  We discuss some
open problems in this vein at the end of the paper.

Let $G$ be a graph with a cellular embedding on a surface of genus $g$
(either orientable or not).  We start with a shortest non-separating
cycle $C$ and generate an ordered family of disjoint cycle sets $\cal
C$ each of which is homologous to $C$.  We use this family to build
another non-separating cycle $C'$ formed by one vertex from each set
in $\cal C$.  Since $C$ is shortest, $C'$ acts a witness giving a
lower bound on $|\cal C|$.  Overall, this gives a lower bound on the
number of edges in $\cal C$, and so an upper bound on $|C|$.
  
We appeal to a combinatorial embedding of the graph which gives, for
each vertex $v$, a clockwise ordering of the edges incident to $v$ as
they are embedded around $v$~\cite{Edmonds60,Youngs63}.  We note that any such embedding can be
maintained under operations such as contraction, deleting, or cutting
along a cycle, via appropriate unions, deletions, or duplications of
the vertex lists which maintain the clockwise orderings; full details are described by Mohar and Thomassen~\cite{mt-gs-01-ch4}.

In the following $\partial f$ denotes the boundary of face $f$.

\begin{lemma} \label{lem:homol}
  Let $G$ be a graph with a cellular embedding on a surface $\cal S$, either orientable or non-orientable.  Let $\cal F$ be
  a set faces of $G$.  We can add a set $L$ of edges to $G$ such that
  \begin{itemize}
  \item $L$ can be incorporated into the embedding of $G$ in a noncrossing way.
  \item The endpoints of $L$ are the set of vertices at distance one from the
    boundaries of $\cal F$.
  \item $L$ decomposes into a set of cycles that is homologous to the boundaries of
    $\cal F$.
  \end{itemize}
\end{lemma}

\begin{proof}

  For a face $f \in {\cal F}$, let $\partial f$ denote the cycle in
  $G$ giving $f$'s boundary, taken in clockwise order.  Let $X$ be the
  set of vertices at distance 1 from $\cal F$ in $G$.  If $f, g \in
  {\cal F}$ are adjacent in $G$ (that is, there is an edge $uv$ such
  that $u \in \partial f$ and $v \in \partial g$ or $f$ and $g$ share
  a vertex $x$), then the vertices at distance 1 from $\partial f$
  interferes with $\partial g$.  To avoid this, we merge adjacent
  faces.  If $f$ and $g$ share a vertex $x$, we cut open the graph at $x$, merging the interiors of $f$ and $g$ and creating two copies of $x$, both on the boundary of the newly created face.  If $f$ and $g$ are connected by an edge $uv$, we cut open the graph along $uv$, merge the interior of
  $f$ with that of $g$ resulting in face $h$.  The edge $uv$ is
  duplicated and both copies appear in $\partial h$.  We repeat this
  operation {\em minimally} until the distance between every pair of faces is at least
  2: that is, performing a sequence of such operations will guarantee that the interior of the resulting faces are homeomorphic to a disk. 
 (Note that on a non-orientable surface, this minimality avoids the possibility that the union of neighboring faces spans a M\"obuis strip, and so the interior remains a topological disk.) 
   Let ${\cal F}'$ be the resulting set of faces and $G'$ the
  resulting graph.  Note that the boundaries of $\cal F$ are
  $\mathbb{Z}_2$-homologous to the boundaries of ${\cal F}'$, since
  the introduction of $uv$ twice cancels under $\mathbb{Z}_2$
  homology.  Note further that interior of each face in ${\cal F}'$ is
  homologous to a disk and thus the boundaries are contractible, and
  the set of vertices at distance one from ${\cal F}'$ is still $X$,
  the set of vertices at distance 1 from ${\cal F}$.

  Let $G''$ be the graph obtained by contracting the boundaries of the
  faces of ${\cal F}'$.  Let $F'$ be the vertices resulting from these
  contractions.  Note again that the set of vertices at distance 1 from
  $F'$ in $G''$ is still $X$, since each vertex at distance 1 from $F'$ must also be within distance 1 of some vertex in $\partial F$, and vice versa.

We will build a cycle that is homologous to each face in ${\cal F}'$ whose vertices are among $X$.  Since the faces in ${\cal F}'$ are at distance at least two from each other, the cycles we construct will not interact with each other.

 \begin{figure}[ht]
    \centering
    \includegraphics[width=3in]{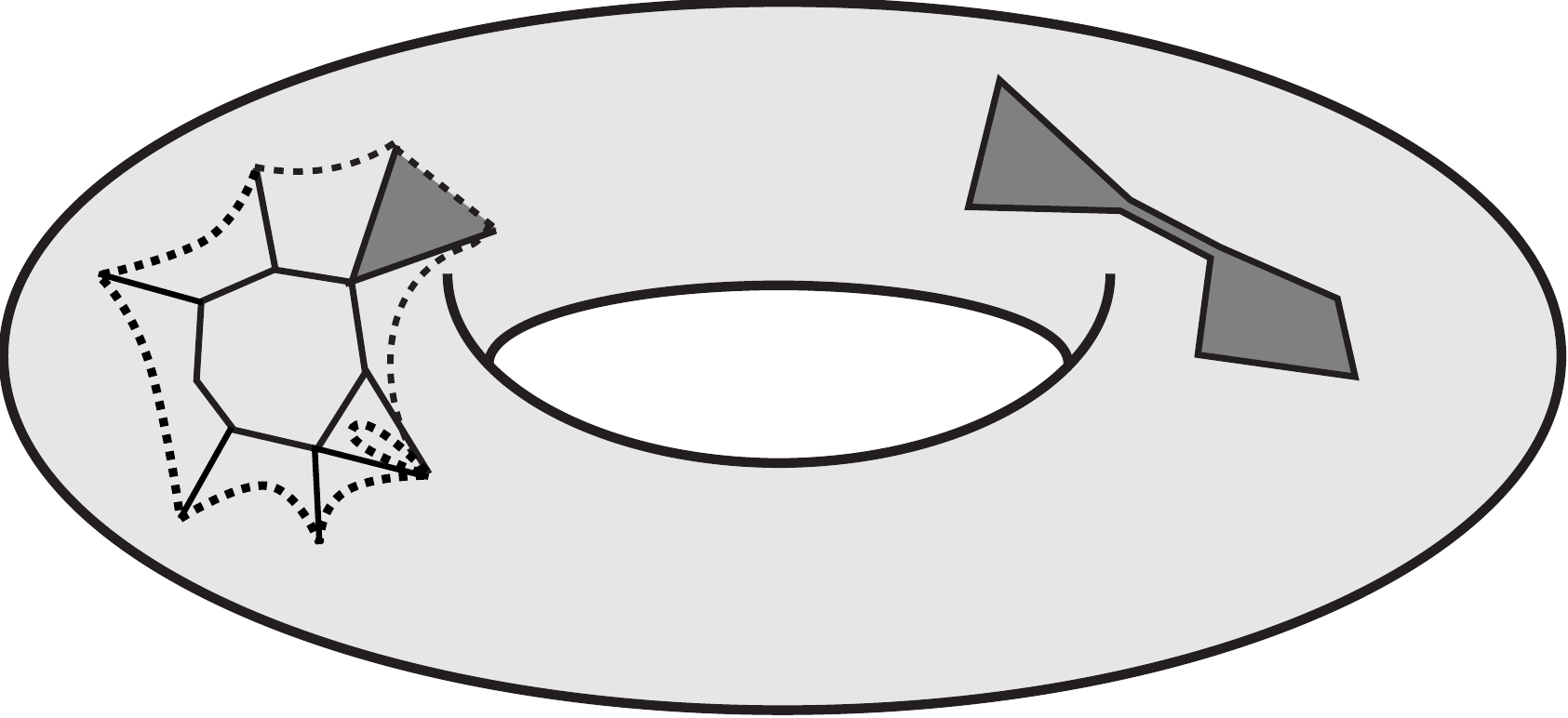}
    \caption{Cycles which are connected by an edge are merged into a single face (shaded above, right), and level edges (shown dashed above) are embedded so that the boundary of the face, incident edges and new edge bounds a topological disk (shaded above, left).}
    \label{fig:cycles}
  \end{figure}

Subdivide every self-loop $\ell$ adjacent to a vertex
in ${\cal F}'$ into two edges  with a vertex $v_\ell$.  Let $G'''$ be the resulting graph.  The set of vertices at distance 1 is now $X'$, which consists of vertices from $X$ and vertices which came from loop subdivisions.

For each vertex $f \in F'$, consider the cyclic clockwise ordering of the
edges incident to the vertex corresponding to $f$ in the embedding of $G'''$.  For every two
consecutive edges $fu$ and $fv$ in this order we introduce the edge $uv$ and call it
a {\em level edge}.  Edge $uv$ can be embedded to be arbitrarily close to $fu$ followed by $fv$; on the original surface, this corresponds to a path following the edge from $u$ to the face $f$, followed by a (possibly empty) portion of the face boundary $\partial f$, followed by the edge from $f$ to $v$; see Figure~\ref{fig:cycles}.
Let $L$ be the set of all such
edges.   Since each such edge can be embedded as described to follow two adjacent edges in the clockwise ordering around the vertex $f$, $G''' \cup L$ can be embedded in a non-crossing way.  Note that self-loops and parallel edges may be introduced this way, e.g.\ when a vertex $f \in F$ has degree 1 or 2, respectively.  See Figure~\ref{fig:loopleveledges}.

The level edges corresponding to $f$ inherit a cyclic ordering from
the ordering of the edges adjacent to $f$.  That is, $uv$ and $vw$ are
consecutive in this ordering if $fu,fv,fw$ are consecutive in the
ordering of edges adjacent to $f$. Further, given how we have embedded $uv$, we know that the cycle $\partial f$ union the edges $fu, uv, fv$ bounds a topological disk.
This implies a partitioning of $L$ into a set of
cycles ${\cal C}$ that is homologous to the boundaries of ${\cal F}'$: simply replace each portion of a face $\partial f$ with the path $fu, uv, fv$.  Since we are (in $\mathbb{Z}_2$ homology sense) adding a set of disks to a cycle, each new cycle is homologous to the original.  This proves the second and third implications of Lemma~\ref{lem:homol}.

However, the endpoints of $L$ are not necessarily vertices of $G$,
since they include the subdividing vertices.  Refer to Figure~\ref{fig:loopleveledges}.  Consider such a vertex
$v_\ell \in X'$ which was used to subdivide self-loop $\ell$.  Merge
any two consecutive edges $uv_\ell$, $v_\ell w$, creating edge $uw$
and minimally modify the embedding so that $uw$ does not intersect
$\ell$.  This maintains the second and third implications.
If there are parallel loops (either on an oriented or non-oriented surface), the connecting level edges consist of bigons between loop vertices; these bigons are null-homologous and hence can be disregarded. The set of
level edges may also have included a self-loop centered at a subdividing
vertex, the new ``edge'' will no longer have any endpoints.  This
``edge'' must bound a topological disk, since, if we introduced a
level edge centered at $v_\ell$, $\ell$ must have bounded a face in
$G''$.  Therefore, we can remove this ``edge'' while maintaining the same homology type for our set of cycles.  We let $L'$ be the modified and remaining edges.  These are the edges satisfying the three implications of Lemma~\ref{lem:homol}. \qed

% Consider such a vertex $v_\ell \in X'$ which was used to subdivide self-loop $\ell$.  Merge the ends of any 

% TODO: GOT STUCK AGAIN HERE  HERE I AM

% The self loop at $f$ adds two edges to the clockwise ordering of edges around $f$, and so there can be anywhere from zero to four edges adjacent to
% $v_\ell$ in this ordering.  If there are no other edges incident to $f$, then
% $\ell$ must be a self-loop;
% remove the resulting self-loop level edge from $L$, whose interior we
% can assume is homologous to a disk.  Similarly, if all of the edges incident to $f$ are loops, 
% each consecutive pair of the loops bounds a topological disk since we know the original graph was a combinatorial embedding.  We can thus remove all the edges between loop vertices, andf each will bound a disk on the interior.  

% Now consider the case where $f$ has both loops and edges leading to vertices from $X$.
% Consider a pair of level edges $v_\ell x$ and $v_\ell y$ (where $x,y \in X$), such
% that the clockwise ordering of edges around $v_\ell$ includes the
% suborder $v_\ell f, v_\ell x, v_\ell y, v_\ell f$, remove $v_\ell x$
% and $v_\ell y$ from $L$ and add $xy$, whose embedding can follow the edges $fx$, $\partial f$, and $fy$ and therefore will also bound a disk.  Note that we must allow double edges here, since two loops with vertices $x$ and $y$ between will result in a loop  consisting of two copies of the edge $xy$; this trivial loop is a part of the level set and must also bound a disk, since the original embedding had every face bounding a disk; see Figure~\ref{fig:loopleveledges}.

 \begin{figure}[ht]
    \centering
    \includegraphics[width=1.5in]{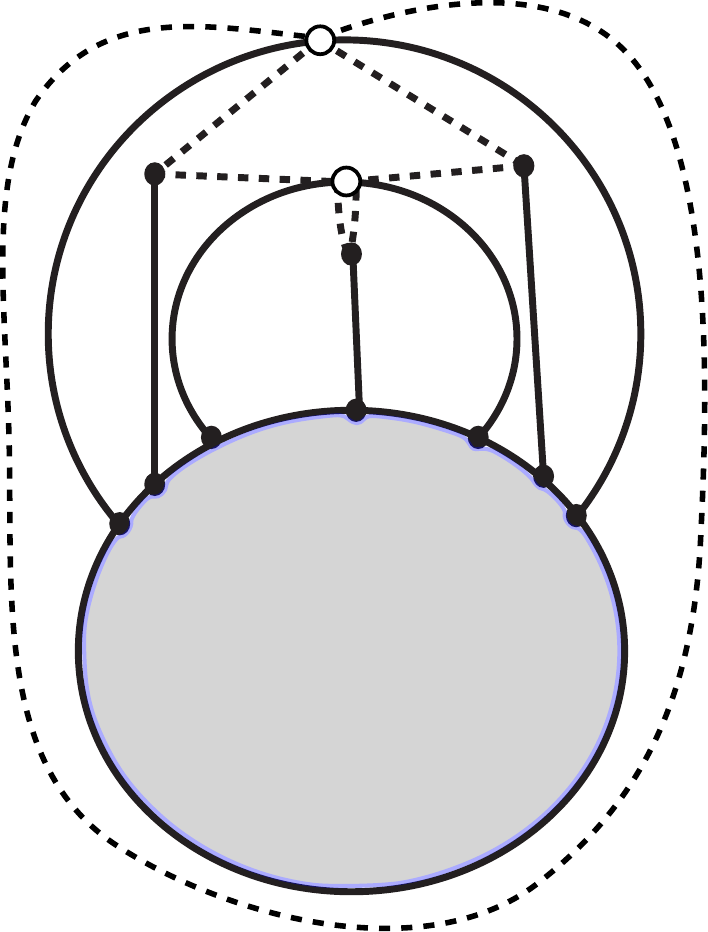}
    \includegraphics[width=1.5in]{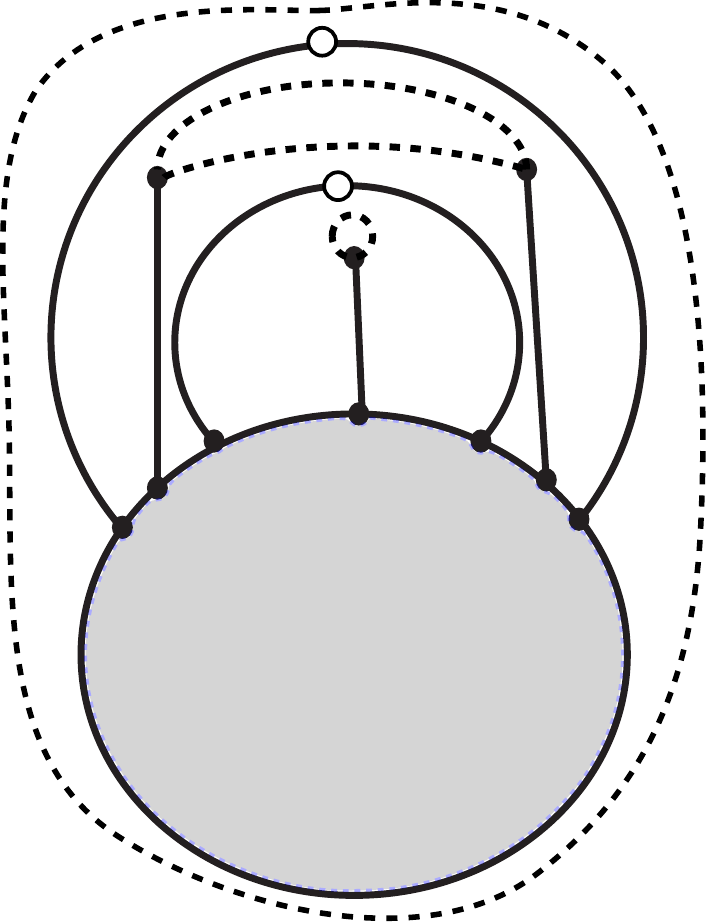}
    \caption{A face $f$ (shaded) and its incident edges (solid) with added subdividing vertices (hollow).  Left: the level edges $L$ (dashed) added to $G'''$.  Right: the level edges after connections to the subdividing vertices are removed.  Note that the outer endpoint-less ``edge'' is not included in $L'$.}
    \label{fig:loopleveledges}
  \end{figure}

% At the end of this process, we have removed all of the loop vertices and added edges between vertices in the set $X$; let $L'$ be the resulting set of edges.  Note that there is still a
% partitioning of $L'$ into a set of cycles ${\cal C}'$ that is
% homologous to $\cal C$, since we can add the cycles described above that are bounded by 
% $\partial f$ plus edges to $X'$ plus the edges of $L$.  Since each such cycle bounds a topological 
% disk and each edge from a face to a vertex of $X$ is used exactly twice, the result is a set of 
% cycles contained in $L'$ that is $\mathbb{Z}_2$-homologous to the boundaries of
% ${\cal F}'$.

\end{proof}

\subsection{Short non-separating cycles}

We are now ready to prove Theorem~\ref{thm:cycle}.

Let $C$ be the shortest non-separating cycle of $G$.  Cut open the
surface and graph along $C$, duplicating $C$ into copies $C_0$ and
$C_0'$; let $G_0$ be the cut open graph.  Glue a disk onto each hole
left from cutting open the graph.  $C_0$ and $C_0'$ are now the
boundaries of faces.

Let $V_i$ be the set of vertices in $G_0$ that are at distance $i$
from $C_0$ and let $s$ be the smallest index such that $V_s \cap
V(C_0') \ne \emptyset$.  We define sets of cycles $C_i$ in a graph
$G_i$, $i = 0 \ldots, s$, starting with $C_0$, inductively as follows:
Given the set of cycles $C_{i-1}$ that are the boundaries of faces
(and starting with $C_0$ as our initial cycle), we define $C_i$ to be
the homologous set of cycles going through $V_i$ as guaranteed by
Lemma~\ref{lem:homol}.  We remove the edges and vertices of $C_{i-1}$
and the edges adjacent to $C_{i-1}$ to make $C_i$ the boundaries of
faces.

  % We add edges to $G'$ called {\em level edges} such that for every $i
  % = 0, \ldots, s$, the set of level edges $C_i$ whose endpoints are in
  % $V_i$ form a set of edge-disjoint cycles that are homologous to
  % $C_0$. We define these edges inductively.  We say that the edges of
  % $C_0$ are level edges.  For each cycle $C \in C_i$, consider the
  % vertices in $V_{i+1}$ that are adjacent to vertices in $C$ in the
  % following cyclic order given by the combinatorial embedding: cut the
  % graph and surface along $C$, contract $C$ to a vertex $v$, delete
  % any self-loops incident to $v$ and use the order implied by the
  % combinatorial embedding around $v$.  Note that this may actually be
  % a multi-order in that a vertex can appear twice in the order.  For
  % every two vertices $a,b$ such that $a$ follows $b$ in the cyclic
  % order (and such that $b$ does not follow $a$), we say that $ab$ is a
  % level edge.  $C_{i+1}$ is the set of all such level edges and forms
  % a set of cycles.  See Figure~\ref{fig:cycles}.  By construction,
  % every level edge can be embedded in a face of $G'$ such that no two
  % level edges in $\cup_{i = 0}^sC_i$ cross.

  For any chord $uv$ of a face $f$, let $P_{uv}$ be the shortest
  $u$-to-$v$ path along the boundary of $f$ and let $\ell(uv) =
  |P_{uv}|$.  Gazit and Miller~\cite{GaMi90} show that for a face $f$
  and a set of pairwise non-crossing chords $H$ across $f$, $\ell(H) \le \frac{1}{8}|f|^2$.
  Since the edges of $\cup_{i = 1}^s C_i$ are chords of the faces of $G$, we get
  \begin{equation}
    \label{eq:level-wts}
    \sum_{i=0}^s \ell(C_i) \le {1\over 8} (||G||_f)^2
  \end{equation}

  By construction $C_i$ is homologous to $C_0$ and so to $C$.  Let
  $\bar C_i$ be the set of cycles obtained from $C_i$ by replacing
  each edge $uv \in C_i$ with $P_{uv}$.  We get $|\bar C_i| =
  \ell(C_i)$.  Since $\bar C_i$ is homologous to $C$, $\bar C_i$ must contain a non-separating cycle $S$.  Since $C$ is the shortest non-separating cycle,
  \[
  |\bar C_i| \ge |S| \ge |C|
  \]

  \begin{figure}[hb]
    \centering
    \includegraphics{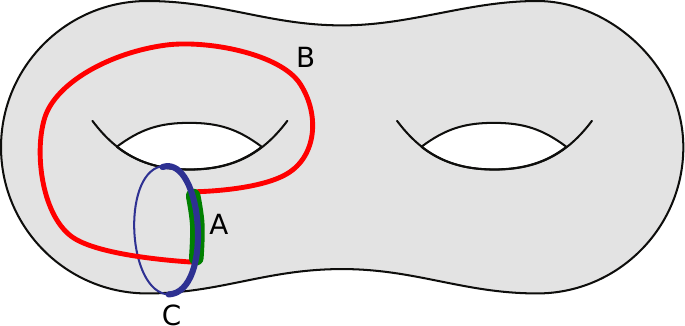}
    \caption{$C$ is the shortest non-separating cycle, $A \cup B$ is another non-separating cycle.}
    \label{fig:ABC}
  \end{figure}

  Let $B$ be the shortest path from $C_0$ to $C_0'$.  Let $A$ be the
  shortest subpath of $C$ that connects $B$'s endpoints.  $A \cup B$
  is a non-separating cycle.  See Figure~\ref{fig:ABC}  Since $|B| = s$ and $|A \cup B| \ge
  |C|$, $s \ge |C|/2$.  We have:
  \[
  {1\over 8}(||G||_f)^2 \ge \sum_{i=1}^s\ell(C_i) =  \sum_{i=1}^s |\bar C_i| \ge |C|^2/2 
  \]
  Rearranging gives Theorem~\ref{thm:cycle}.

\subsection{Planarizing sets}

Repeatedly cutting along non-separating cycles allows us to reduce a surface-embedded graph to a planar graph, while only reducing the face norm:

\begin{lemma}\label{lem:cut-n-contract}
  Let $G$ be an embedded graph and let $C$ be a non-separating cycle.
  Cutting open the graph along $C$ and then contracting each resulting
  copy of $C$ results in a graph $G'$ such that $||G'||_f <
  ||G||_f$.
\end{lemma}

\begin{proof}
  Let $\cal F$ be the set of faces of $G$ and let ${\cal F}_C$ be the
  set of faces of $G$ that have a bounding edge in $C$.  Cutting along
  a non-separating cycle $C$ of a graph $G$ embedded on surface $\cal
  S$ introduces two holes, each bounded by a copy of $C$.  Contracting
  each hole and each bounding copy of $C$ results in a graph $G'$ with
  face set ${\cal F}'$.  Every face in ${\cal F}$ maps to a face in
  ${\cal F}'$ such that the faces in ${\cal F} \setminus {\cal F}_C$
  are the same size as their image in ${\cal F}'$ and the faces in
  ${\cal F}_C$ are strictly larger than their counterparts ${\cal F}'_C$ in ${\cal
    F}'$ giving:
  \[
  ||G'||_f = \sqrt{\sum_{f \in {\cal F}'} |f|^2 }
 = \sqrt{ \sum_{f \in {\cal F}'\setminus{\cal F}'_C} |f|^2 +\sum_{f \in {\cal F}'_C} |f|^2 } < \sqrt{ \sum_{f \in {\cal F}\setminus{\cal F}_C} |f|^2 +\sum_{f \in {\cal F}_C} |f|^2 } = \sqrt{\sum_{f \in {\cal F}} |f|^2 }= ||G||_f \] \qed
\end{proof}

Cutting along a non-separating cycle $C$ of a graph $G$ embedded on
surface $\cal S$ reduces the genus of the surface by one and
introduces two holes, each bounded by a copy of $C$.
Lemma~\ref{lem:cut-n-contract} shows that if we contract the two
copies of $C$ (and the corresponding holes), we only reduce the
face-norm of the graph.  We can repeat this cut-and-contract procedure
$g$ times, each time we find a non-separating cycle of length at most
${1\over 2}||G||_f$, at which point the surface is a sphere and the
final graph $G'$ is planar.  Of course, applying this method to the
dual $G^*$ of the graph, results in a set of {\em planarizing} edges
whose size is measured in terms of the vertex-norm
\[
||G||_\delta = \sqrt{\sum_{v \in V} \delta(v)^2}
\]
of $G$ where $\delta(v)$ is the degree of vertex $v$.  Recall that the
dual of a plane graph is given by a vertex for every face of the primal
graph, with dual vertices connected when the corresponding primal
faces are adjacent.  By duality, the degree of a vertex is the size of
the face in the dual corresponding to the vertex.  We get:
\begin{lemma}\label{lem:planarize}
  There is a set of $\frac{g}{2}||G||_\delta$ edges of a genus-$g$
  graph whose removal leaves a planar graph.
\end{lemma}

\section{Pairwise-crossing number of surfaces}\label{sec:cross}

There are many measures of how close a graph is to being planar.  One measure is the crossing number which is the minimum number of edge crossings in a planar, topological drawing of the graph~\cite{Turan77}.  A drawing is {\em topological} if vertices map to distinct points and edges map to simple Jordan arcs connecting the points their endpoints such that (i) no arc passes through a vertex different from its endpoints, (ii) no two arcs meet in more than one point, and (iii) no three arcs share a common interior point.  Formally the crossing number of a fixed drawing is number of interior points that are shared by two arcs.  The restriction to topological drawings does not increase the crossing number of a graph, see e.g.~\cite{Felsner04}.  Rather than planar drawings, we are interested in drawings on surfaces of genus $g$ and so will refer to {\em surface topological drawings}.  This number has been studied by Shahrokhi, Sz\'{e}kely, S\'{y}kora and Vrt'o, who give upper and lower bounds on the crossing number of complete graphs drawn on compact 2-manifolds~\cite{SSSV94,SSSV96}; more specific bounds are also known for surfaces such as the torus~\cite{Guy1968376}.

We first use the crossing number of a particular drawing of a graph to give bounds on the size of a set of edges whose removal results in  a topological drawing in the plane.

\begin{lemma}\label{lem:topo-planarize}
A graph $G$ admitting  a topological drawing on a surface $\cal S$ of genus $g$ with $\chi$ crossings has a subset of at most 
\[ {g \over 2}\sqrt{16\chi+||G||_\delta^2}\] edges whose removal
leaves a graph whose inherited drawing is a planar topological
drawing.
\end{lemma}

\begin{proof}Let $H$ be the graph embedded on $\cal S$ obtained
from $G$ by introducing a vertex at each crossing.  Since the drawing is topological, each of these new vertices has degree 4.  We have that $||H||_\delta^2 = \sum_{v\in H} \delta_H(v)^2 = 16\chi+\sum_{v\in G} \delta_G(v)^2$.  By Lemma~\ref{lem:planarize}, $H$ has a planarizing edge set $S_H$ with at most ${g\over 2}||H||_\delta$ edges.  Let $S_G$ be the set of edges of $G$ from which $S_H$ are generated.  Since $|S_G| \le |S_H|$, the lemma follows.
\hfill \qed \end{proof}

Another class of graphs that is close to being planar are the class of
{\em $k$-quasi-planar graphs}.  A graph is $k$-quasi-planar if it
admits a planar, topological drawing in which no subset of $k+1$ edges
pairwise cross; thus a graph that is 1-quasi-planar is planar.
Various bounds on the number of edges in such graphs have been
given~\cite{AAPPS97,PSS94,FPS13}, culminating in:
\begin{theorem}[Suk and Walczak~\cite{SW13}]\label{thm:planar-cross}
  A simple $n$-vertex graph admitting a topological drawing in the
  plane in which no subset of $k+1$ edges pairwise cross has at most
  $c_k n \log n$ edges where $c_k$ is a constant depending only on $k$.
\end{theorem}
In fact, if one
follows the dependence on $k$ through Suk and Walczak's work, one finds that
\begin{equation}
c_k = A^{k^k} \text{ for a fixed constant }A \ge 2\label{eq:ck}
\end{equation}
As far
as we know, such bounds have not previously been generalized to more
general surface topological drawings as we do so here.  The proof of Theorem~\ref{thm:cross} is based on the analysis technique of Pach
et~al.~\cite{PSS94}, but here we are able to immediately reduce the genus $g$ topological graph to a planar, topological graph, thus invoking Suk and Walczak's result~\cite{SW13}.

\begin{theorem} \label{thm:cross} A simple $n$-vertex graph admitting
  a topological drawing on a surface of genus $g > 0$ in which
  no subset of $k+1$ edges pairwise cross has at most 
  $(2g^2)^k c_k n \log n$ when $g = O(n)$.
\end{theorem}

\begin{proof}
  Let ${\cal G}_{g,k,n}$ be the family of all graphs with at most $n$
  vertices and admitting a genus-$g$ topological drawing in which no
  subset of $k+1$ edges pairwise cross.  Let $m_{g,k,n}$ be the
  maximum number of edges in any graph in ${\cal G}_{g,k,n}$.

  We aim to prove the assertion for ${\cal G}_{g,k,n}$ that 
  \begin{equation}
    \label{eq:2}
    m_{g,k,n} \le (2g^2)^{k}c_k n \log n
  \end{equation}
  by induction over $k$.  For $k = 1$ (and every $g$ and $n$), the
  assertion is true since such graphs are genus-$g$ graphs and have
  $O(n + g)$ edges which is $O(n)$ for $g = O(n)$.  For values of $n$
  such that $n\log n \le (2g^2)^{k}c_k$, the assertion is
  true since the right-hand side of Inequality~(\ref{eq:2}) exceeds $n^2$ for all
  such values of $n$.  We assume that $m_{g,k-1,n} \le
  (2g^2)^{k-1}c_{k-1} n \log n$.

  Consider a graph $G \in {\cal G}_{g,k,n}$ and fix a genus-$g$ topological
  drawing of $G$ in which no subset of $k+1$ edges pairwise cross.
  Let $\chi$ be the number of crossings in this drawing.  We first bound $\chi$ so we may use
  Lemma~\ref{lem:topo-planarize}.

  Consider an edge $e$ of $G$ and let $G_e$ be the subgraph of $G$
  consisting of all the edges crossing $e$.  Let $G_e$ inherit its
  drawing from $G$.  Since the drawing of $G$ has no $k+1$ pairwise
  crossing edges, the drawing of $G_e$ has no $k$ pairwise crossing
  edges for otherwise such a set along with $e$ would witness a set of
  $k+1$ pairwise crossing edges in the drawing of $G$.  Therefore $G_e
  \in {\cal G}_{g,k-1,n}$ and so $G_e$ has at most $m_{g,k-1,n}$
  edges.  The number of crossings on $e$ is therefore at most
  $m_{g,k-1,n}$.  Summing over all edges of $G$, $\chi \le
  \frac{1}{2}m \cdot m_{g,k-1,n}$ where $m$ is the number of edges in $G$.  By the inductive hypothesis,
   \begin{equation}
     \label{eq:3}
     \chi \le \frac{1}{2} m\cdot (2g^2)^{k-1}c_{k-1} n \log n.
   \end{equation}
   
   Let $S$ be the set of edges forming a planarizing set for $G$ guaranteed by Lemma~\ref{lem:topo-planarize}.  By
   Lemma~\ref{lem:topo-planarize}, Equation~\eqref{eq:3} and the fact
   that $||H||_\delta^2 \le 2|E(H)|\cdot|V(H)|$ for any graph $H$,
   \begin{equation}
     \label{eq:4}
     |S| \le \frac{g}{2} \sqrt{8 m\cdot (2g^2)^{k-1}c_{k-1} n \log n+2mn} \le \frac{3g}{2} \sqrt{ m\cdot (2g^2)^{k-1}c_{k-1} n \log n}
   \end{equation}
   where the last inequality holds for $n$ such that $2 < (2g^2)^{k-1}c_{k-1}\log n$; these coincide with non-base-case values of $n$. 
    Let $G'$ be
   the graph obtained by deleting $S$ from $G$.  Then $m \le E(G') +
   |S|$.  Since $G'$ is a $k$-quasi-planar graph on at most $n$
   vertices, $|E(G')| \le c_k n \log n$ by
   Theorem~\ref{thm:planar-cross}.  Combining, we get
   \begin{equation*}
     m \le c_k n \log n + \frac{3g}{2} \sqrt{ m\cdot (2g^2)^{k-1}c_{k-1} n \log n}
   \end{equation*}
   Rearranging:
   \begin{equation} \label{eq:f}
     m - \frac{3g}{2}\sqrt{(2g^2)^{k-1}c_{k-1} n \log n} \sqrt{m} \le c_k n \log n 
   \end{equation}
   Let $f(m) = m - \frac{3g}{2}\sqrt{(2g^2)^{k-1}c_{k-1} n \log n}
   \sqrt{m}$. We consider the two cases corresponding to the sign of
   the left-hand side of (\ref{eq:f}).

   If $f(m) \le 0$, then
   \[m \le \left(\frac{3g}{2}\right)^2 (2g^2)^{k-1}c_{k-1} n \log n =
   (2g^2)^k \frac{9}{8}c_{k-1} n \log n \le (2g^2)^kc_k n \log n,\]
   where the last inequality follows from ${9 \over 8} c_{k-1} < c_k$
   (which is clearly true given Equation~(\ref{eq:ck})), thus proving
   the assertion.
   
   We note that $f(m)$ is an increasing function for all positive
   values of $m$ such that $f(m) > 0$.  
%\marginpar{Erin: How obvious is that it is increasing?  Justify?}
We will show that
   \begin{equation}
f((2g^2)^{k}c_k n \log n) > c_{k} n \log n, \label{eq:ff}
\end{equation}
implying that $m < (2g^2)^{k}c_k n \log n$ when $f(m) > 0$, proving
the assertion.
\begin{eqnarray*}
  f((2g^2)^{k}c_k n \log n) &=& (2g^2)^{k}c_k n \log n -\frac{3g}{2}\sqrt{(2g^2)^{k-1}c_{k-1} n \log n} \sqrt{(2g^2)^{k}c_k n \log n} \\
  & = & (2g^2)^{k}c_k n \log n - \frac{3\sqrt{2}}{4}\sqrt{(2g^2)^{2k}c_{k-1}c_{k}}n\log n \\
  & = & (2g^2)^{k}c_k n \log n  \left(1-\frac{3\sqrt{2}}{4}\sqrt{\frac{c_{k-1}}{c_k}}\right)\\
  & > & (2g^2)^{k}c_k n \log n  \left(1-\frac{3\sqrt{2}}{4}\frac{1}{\sqrt{2}}\right)\text{, since $c_k > 2 c_{k-1}$, by Equation~(\ref{eq:ck})}\\
  & = & (2g^2)^{k}c_k n \log n  \left(\frac{1}{4}\right)\\
  & > & c_k n \log n \mbox{, for $k \ge 2$ and $g \ge 1$}
 \end{eqnarray*}
This proves Equation~(\ref{eq:ff}) and so the theorem.
\hfill \qed \end{proof}

\section{The $(p_g,q_g)$-property of genus-$g$ ball systems}\label{sec:pq}

The proof of the fact that the ball system of a graph of genus $g$ has the $(p_g,q_g)$-property is similar to the proof of Proposition 2 in the work of Chepoi, Estellon and Vax\`{e}s~\cite{planarballcover}, although we have made efforts to simplify the proof here.

Let $G$ be a graph of diameter at most $2R$ with an embedding on a
surface $\cal S$ of genus $g$.  Let $C$ be a set of $p_g$ vertices; we
will define $p_g$ shortly.  Consider a set of shortest paths ${\cal P}
= \{P_{ij}\ : \ c_i, c_j \in C\}$ where $P_{ij}$ is the shortest
$c_i$-to-$c_j$ path in $G$.  We can assume, without loss of
generality, that the intersection of any two of these paths is {\em simple}, having at
most one component (a path or vertex), for otherwise, one path could
be redirected along another without compromising shortness as
illustrated in Figure~\ref{fig:simple}.

\begin{figure}[t]
  \centering
  \label{fig:simple}
  \includegraphics{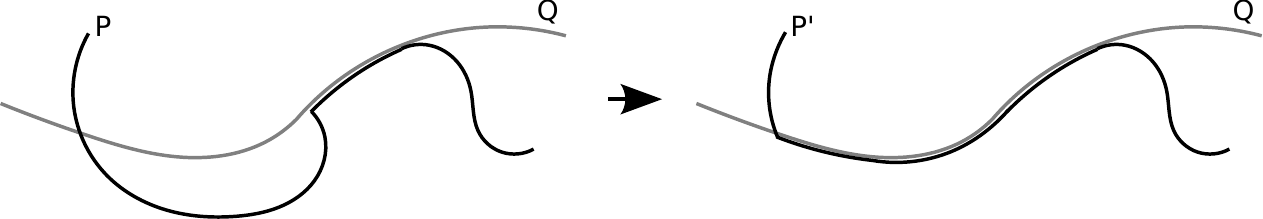}
  \caption{If $P$ and $Q$ are both shortest paths, then $P'$ must also be a shortest path between $P$'s endpoints.}
\end{figure}

 Taking the image of $P_{ij}$ on the surface for each path $P_{ij} \in {\cal
 P}$, we get a drawing of the complete graph $K_{p_g}$ on $\cal S$.  We can
 make this drawing topological by a sequence of simple, local
 transformations, as illustrated in Figure~\ref{fig:simplify}. Since we
 assumed that path intersections are simple, the first
 transformation modifies the drawing to achieve the first and third
 properties of a topological drawing and the second transformation
 modifies the drawing to achieve the second property of a topological
 drawing.  These transformations {\em respect intersection} so far as
 that, in the final drawing of $K_{p_g}$, the images of two edges of
 $K_{p_g}$ share a point if and only if the corresponding paths share
 a vertex in $G$.

\begin{figure}[t]
  \centering
  \label{fig:simplify}
  \includegraphics{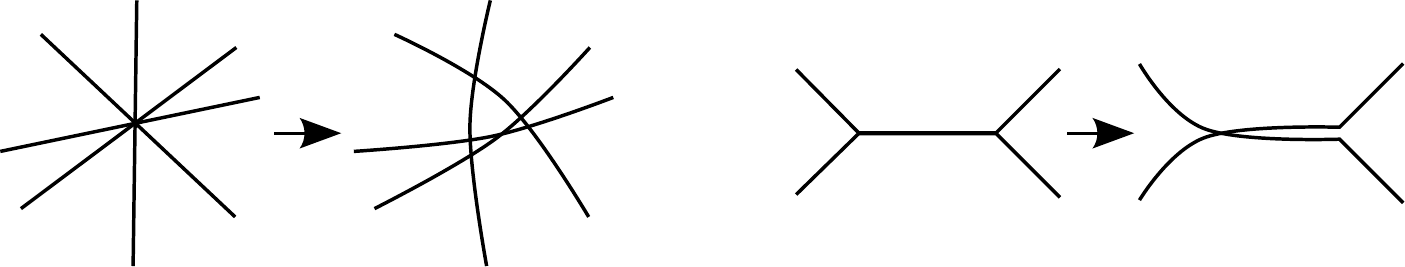}
  \caption{Making a drawing topological; full details are given by Felsner~\cite{Felsner04}.}
\end{figure}

Since the drawing of $K_{p_g}$ is a topological drawing on surface $\cal S$
of genus $g$, we can use Theorem~\ref{thm:cross} to guarantee that, for $p_g$ sufficiently large (and depending only on $g$), this drawing contains a subset of at least
$2q_g-3$ edges that pairwise cross.  Likewise, since
the drawing of $K_{p_g}$ respects intersections, there must be a
subset ${\cal P}'$ of at least $2q_g-3$ paths of $\cal P$ that pairwise
intersect.  We pick the {\em midpoint} of a $c_i$-to-$c_j$ path
$P_{ij} \in {\cal P}$ to be any vertex $m_{ij}$ that is in $B(c_i)\cap
B(c_j)$; since diameter of the graph is at most $2R$, the paths $\cal
P$ are shortest and the balls have radius $R$, such a point always
exists.

\begin{claim}
  For any two paths $P_{ij}, P_{k\ell} \in {\cal P}'$, either $m_{ij} \in B(c_k)\cup B(c_\ell)$ or $m_{k\ell}
\in B(c_i) \cup B(c_j)$.
\end{claim}

\begin{proof}
  Let $x$ be a vertex shared by both $P_{ij}$ and $P_{k\ell}$.
  Assume, w.l.o.g., that $c_i$ is the closest of the endpoints of
  $P_{ij}$ and $P_{k\ell}$ ($\{c_i,c_j,c_k,c_\ell\}$) to $x$.  Also
  assume, w.l.o.g., that $x$ is in the $c_k$-to-$m_{k\ell}$ subpath of
  $P_{k\ell}$.  Since $m_{k\ell} \in B(c_k)$ the distance from
  $m_{k\ell}$ to $x$ to $c_k$ is at most $R$ and since $c_i$ is closer
  to $x$ than $c_k$, then the distance from $m_{k\ell}$ to $x$ to
  $c_i$ is also at most $R$, therefore $m_{k\ell} \in B(c_i)$.
\hfill \qed \end{proof}

Since this claim holds for every pair of paths, by an averaging
argument, there must be some path $P_{ij}$ whose midpoint is contained
in the ball centered at the endpoint of at least $\frac{1}{|{\cal
    P}'|} \cdot {|{\cal P}'| \choose 2} = \frac{1}{2}(|{\cal P}'|-1)
\ge q_g-2$ paths.  Since $m_{ij}$ is additionally contained in $B(c_i)
\cap B(c_j)$, $m_{ij}$ is a point contained in $q_g$ balls, showing
that the ball system for $G$ has the $(p_g,q_g)$-property.

\section{Handling apices and toward minor-excluded graphs} \label{sec:apex}

The Graph Minor Structure Theorem is one of many results of Robertson
and Seymour leading to the Graph Minor theorem.  The Graph Minor
Structure Theorem shows that for a fixed graph $H$, any graph
excluding $H$ as a minor is composed of graphs that, after the removal
of a fixed number of vertices, can be embedded on a surface in which
$H$ cannot be embedded with a fixed number of {\em vortices}
(described below).  These subgraphs are glued together in a tree-like
structure called a tree decomposition.

We are able to show that we can {\em remove} the apices, so to speak,
of one of these graphs (Section~\ref{sec:rem-ap}) and that there is
one subgraph within which every other vertex is distance $R$ 
(Section~\ref{sec:central}).  Of course, this subgraph may be quite
large, but since this subgraph is nearly embeddable on some surface,
after removal of the apices, it may be possible to use arguments
similar to those in Section~\ref{sec:pq}.  We discuss this further in
Section~\ref{sec:mf-decomp}.

\subsection{Removing apices} \label{sec:rem-ap}

We show something stronger than that of {\em removing apices from
  bounded genus graphs}:

\begin{lemma}\label{lem:apex-removal}
  Let $\cal G$ be a class of graphs whose ball systems have
  VC-dimension at most $q-1$ and satisfy the $(p,q)$-property.  Then
  there is a constant $\rho$ such that any graph in $\cal G$ with an
  additional $\alpha$ apices and diameter at most $2R$ can be covered
  by at most $\rho+\alpha$ balls of radius $R$.
\end{lemma}

\begin{proof}
  Let $G$ be a graph such that for a subset of at most $\alpha$ vertices $A$, $G
  \setminus A \in {\cal G}$.  Let $\cal B$ be the ball system for $G$
  and let ${\cal B}'$ be the subset of those balls that do not
  intersect $A$.  The VC-dimension of ${\cal B}'$ is at most that of
  ${\cal B}$, which is at most $q-1$.  Likewise, since $\cal B$ has
  the $(p,q)$-property, so does ${\cal B}'$.  By the Fractional Helly
  Theorem, it follows that ${\cal B}'$ has a hitting set of size at
  most $\rho$; this hitting set along with $A$ is a hitting set for
  $\cal B$. 
\hfill \qed \end{proof}

\subsection{The central node of a tree decomposition} \label{sec:central}

A tree decomposition $\cal T$ of a graph $G = (V,E)$ is a pair $(T,
{\cal X})$ where $T$ is a tree and $\cal X$ is a family of subsets (or
{\em bags}) of $V$ such that:
\begin{itemize}
\item Each node $a$ of $T$ has a corresponding subset $X_a \in {\cal
    X}$ and $\cup_{X \in {\cal X}} X = V$;
\item For every edge $uv \in E$ there is a bag $X \in {\cal X}$ such
  that $u,v \in X$.
\item For any three nodes $a,b,c \in T$ such that $b$ is on the
  $a$-to-$c$ path in $T$, $X_a \cap X_c \subseteq X_b$.
\end{itemize}
We refer to the {\em nodes} of $T$ and {\em vertices} of $G$ to avoid
confusion.  The width of a tree decomposition $(T,{\cal X})$ is
$\max_{X \in{\cal X}} |X|-1$.  Tree decompositions are not unique.
The treewidth of a graph is the minimum possible width of a tree
decomposition of the graph.

We show that given a tree decomposition of a graph of diameter $2R$,
there is a node $a$ of the tree decomposition such that every vertex
in the graph is within distance $R$ of some vertex in $X_a$.  This is
similar to Theorem~5 by Gavoille et~al.~\cite{GPRS01}, but we are
specific about the node of interest in the tree decomposition.  We
include the proof below for completeness.

\begin{theorem}[Central node] \label{thm:central-node}
  There is a node $v$ of a tree decomposition ${\cal T} = (T,{\cal
    X})$ of a graph $G$ with diameter at most $2R$ such that every vertex of
  $G$ is within distance $R$ of some vertex in $X_v$; i.e.\ $d(x,X_v)
  \le R$ for every vertex $x$ of $G$.
\end{theorem}

Consider a node $u$ of $T$ and the corresponding bag $X_u \in {\cal
  X}$. Removing $u$ from $T$ and $X_u$ from $G$ results in $k \ge 1$
subgraphs, each with a tree decomposition derived from ${\cal T}$.
Formally, let $T^1_u, \ldots, T^k_u$ be the components of $T \setminus
\{u\}$.  Let ${\cal X}_u^j$ be the bags corresponding to nodes of
$T_u^j$ with the vertices in $X_u$ removed: ${\cal X}_u^j = \{ X_v
\setminus X_u \ : \ v \in T_u^j\}$.  Let $V_u^j$ be the vertices in
the bags corresponding to nodes of $T_u^j$ with $X_v$ removed: $V_u^j
= \cup {\cal X}_u^i$. ${\cal T}_u^j = (T_u^j, {\cal
  X}_u^j)$ is a tree decomposition of the subgraph of $G$ induced by
$V_u^j$. Since $X_u$ is a vertex separator, any $v$-to-$w$ path in $G$
for $v \in V_u^i$ and $w \in V_u^j$ ($i \ne j$) must contain a vertex
of $X_u$.

Let $d(x,y)$ be the shortest-path distance between $x$ and $y$ in $G$.
For a subset of vertices $Y$, let $d(x,Y)$ be the minimum distance
from $x$ to any vertex of $Y$, so $d(x,Y) = \min_{y \in Y}d(x,y)$. For
any two subsets $X$ and $Y$, let $f(X,Y)$ be the furthest vertex in
$X$ from $Y$; i.e.\ $f(X,Y) = \arg\max_{x \in X} d(x,Y)$.

\begin{lemma}\label{lem:unique}
  If the distance from the furthest vertex in $V^i_u$ to $X_u$ is
  greater than $R$ for any $i$, then for every $j \ne i$, the distance
  from the furthest vertex in $V^j_u$ is strictly less than $R$.
  %If $d_i(X_u) \ge R$ for any $i$, then $d(f(V^j_u,X_u)) \le R$ for all $j \ne i$.
\end{lemma}

\begin{proof}
  Let $f_i = f(V^i_u,X_u)$ and let $f_j = f(V^j_u,X_u)$ for $i \ne j$.

  Let $x$ be a vertex in $X_u$ that is on a shortest path from
  $f_i$ Note also that since
  $f(V^i_u,X_u)$
  and $f(V^j_u,X_u)$ are both vertices in $G$,
  $d(f(V^i_u,X_u),f(V^j_u,X_u)) \le 2R$. So we have:
  \begin{eqnarray*}
    2R & \ge & d(f(V^i_u,X_u),f(V^j_u,X_u)) 
    \\&=& d(f(V^i_u,X_u),x)+d(f(V^j_u,X_u),x)\\
    &\ge& d(f(V^i_u,X_u),X_u)+d(f(V^j_u,X_u),X_u) \\
    &> &R+ d(f(V^j_u,X_u),X_u)
\end{eqnarray*}
  The above then immediately implies that $d(f(V^j_u,X_u),X_u) < R$.
\hfill \qed \end{proof}

  Consider the following procedure for finding the central node, starting
  at an arbitrary node $r$:
  \begin{tabbing}
  1\qquad \=  {\sc search}$(r)$\\
  2 \> \qquad\= If $d(x,X_r) \le R$ for all $x \in V(G)$, return $r$. \\
  3 \> \> Otherwise:\\
  4 \> \> \qquad \= Let $p$ be a node adjacent to $r$ in $T$ such that $d(f(V_r^i,X_r),X_r) > R$ and $p \in T_r^i$. \\
  5 \> \> \> {\sc search}$(p)$.
  \end{tabbing}
  
  It is clear that if this procedure terminates, then the statement of
  the lemma is true.  It remains to argue that the algorithm must
  terminate.  If we we reach line 4, then, by Lemma~\ref{lem:unique},
  $p$ is unique.  If {\sc search} does not terminate, then it is easy
  to see that {\sc search} must oscillate between two adjacent nodes
  $p$ and $q$ of the tree decomposition: {\sc search}$(p)$ calls {\sc
    search}$(q)$ and vice versa.  In this case, there must be a vertex
  $x \in T_q^i$ where $i$ is such that $d(f(V_q^i,X_q),X_q) > R$ and
  $p \in T_q^i$ and a vertex $y \in T_p^j$ where $j$ is such that let
  $d(f(V_p^j,X_p),X_p) > R$ and $q \in T_p^j$.  Let $S$ be a shortest
  $x$-to-$y$ path; by definition of $p$ and $q$, $S$ must visit a
  vertex $a \in X_p$ and a vertex $b \in X_q$ (possibly $a = b$). Let
  $m$ be a vertex closest to the middle of $S$.  Since the diameter of
  $G$ is at most $2R$, $d(x,m)$ and $d(y,m)$ is at most $R$.
  Therefore $b$ must come after $m$ along $S$ from $x$ to $y$ and $a$
  must come after $m$ along $S$ from $y$ to $x$.  It must be that $m =
  a = b$, contradicting that $d(x,X_q) > R$ and $d(y,X_p) > R$. This
  concludes the proof of the Central Node Theorem.

Theorem 5 of Gavoille et~al.'s work is an immediate corollary of Theorem~\ref{thm:central-node}:

\begin{corollary}[Theorem 5~\cite{GPRS01}]
  For a graph with treewidth $tw$ and diameter $2R$, there is a set $S$ of at most
  $tw+1$ vertices such that $d(x,S) \leq R$ for every vertex $x$ in
  the graph.
\end{corollary}

\subsection{Minor-free decompositions} \label{sec:mf-decomp}

Finally, we outline a direction for extending this result to
minor-free graph classes and describe the challenges.

Robertson and Seymour showed that for any graph $G_H$ that excludes a
fixed minor $H$, $G_H$ has a well-defined
structure~\cite{Robertson200343}.  Using the notation and terminology
of Demaine et~al.~\cite{HMinorFree_JACM}, the Graph Minor Structure
Theorem states that $G_H$ is obtained by $h$-clique sums of graphs
that are {\em $h$-almost embeddable} on surfaces in which $H$ cannot
be embedded.  A graph $G$ is
\emph{$h$-almost-embeddable} on a surface $S$ if:
\begin{itemize}
\item There is a set $A$ of at most $h$ vertices, called \emph{apex}
  vertices, such that $G\setminus A$ can be written as a union of graphs $G_0
  \cup G_1 \cup \cdots \cup G_h$ where $G_0$ can be cellularly
  embedded on $S$.

\item For every $i > 0$, $G_i$ is a graph, called a {\em vortex}, that
  has a tree-decomposition that is a path with nodes in order
  $x_i^1,x_i^2, \ldots$ and width at most $h$.

\item For every $i > 0$, there is a face $F_i$ such that $u_i^1,
  u_i^2, \ldots$ is a subset of the boundary vertices of $F_i$ in
  order along the boundary of $F_i$ and $u_i^j \in X_{x_i^j}$ for all $j$.
\end{itemize}
Note that since $H$ is fixed, the surfaces in which the
components of $G_H$ are almost embeddable have fixed genus.

An $h$-clique sum between graphs $A$ and $B$ identifies the vertices
of a clique on at most $h$ vertices in $A$ and $B$ and then possibly
removes some edges of the clique.  The clique-sum of graphs provides a
natural tree decomposition.  Specifically, $G_H$ admits a tree
decomposition $(T,{\cal X})$ such that for every $X \in {\cal X}$, the
subgraph of $G_H$ induced by $X$ is $h$-almost embeddable and the
intersection of any two sets of $\cal X$ contains at most $h$
vertices.  Using this decomposition, we define the {\em central subgraph}
of $G_H$ as the subgraph of $G_H$ induced by the vertices in the
central node of this tree decomposition.

Focussing on this central subgraph, we can {\em remove} the apices by
way of Lemma~\ref{lem:apex-removal}.  Now, in the efforts to prove the
$(p,q)$-property for the set of balls not intersecting apices of the
central subgraph, consider a set of $p$ balls for sufficiently large
$p$.  We can assume w.l.o.g.\ that at most one ball center is in each
of the neighboring $H$-minor-free graphs that are clique-summed to the
central subgraph; if a large number of ball centers are in one
neighbor, then since the balls must all reach the central subgraph, a
large enough number of them must share a vertex, since the clique sums
are small.  

We can then focus on center-to-center shortest paths, as in
Section~\ref{sec:pq}.  For this proof technique, we need to show that
among a set of center-to-center shortest paths, a sufficiently large
number of them share an interior vertex.  While these paths must cross
the central subgraph and parts of them must be embedded on the surface
that the bulk of the central subgraph is embedded on, these paths can
use the clique sums and vortices to {\em hop} over eachother, crossing
without intersecting.  It does not seem possible to bound how much
this can happen since the {\em number} of vortices and clique sums is
not bounded.  so it is likely that a more global argument, taking into
account the balls and not just the shortest paths between ball
centers, will be required in order to illustrate the $(p,q)$-property.

\section{Discussion} \label{sec:future}

This paper presents a generalization of the ball-cover property to bounded genus graphs with a constant number of apices.  
%We conjecture that this holds true for surface-embedded graphs of arbitrary genus; however, in order to show this using our method, a better norm-based separator (with size $O(\sqrt{gn})$ rather than $O(g \sqrt{n})$) would be necessary.  
This represents a significant step towards showing this result holds for all minor-free families of graphs.  This work leaves open this direct question and several others.  

For one, these results, ours and that of Chepoi et~al., do not
evaluate the explicit number of balls required for coverage, relying
as we do, on the Fractional Helly Theorem.  Tracing the constant
through Matou\u{s}ek's work reportedly results in a constant in excess
of 800~\cite{talk} while the best lower bound known is
4~\cite{GPRS01}.  A direct proof, bypassing the Fractional Helly
Theorem, is likely necessary to result in more practical answers.
Likewise, an algorithmic result is desirable, particularly if the
application to interval routing is to be taken seriously.

Further, since our planarizing set (Lemma~\ref{lem:planarize}) reduces
a graph of genus $g$ to a planar graph after the removal of
$O(g||G||_\delta)$ edges and since Gazit and Miller give an
$O(||G||_\delta)$ balanced edge separator for planar graphs, we can
combine these results to get an $O(g||G||_\delta)$ edge separator for
genus-$g$ graphs.  The obvious question is whether an $O(\sqrt{g}||G||_\delta)$, balanced edge separator exists for genus-$g$ graphs.  Much like Gazit and Miller's separator is a strictly tighter bound on size than the pre-existing $O(\sqrt{\delta_{\max} n})$ balanced edge separator for planar graphs~\cite{Miller86,DDSV93}, an $O(\sqrt{g}||G||_\delta)$, balanced edge separator for genus-$g$ graphs would be a strictly tighter bound.  Our implied  $O(g||G||_\delta)$ edge separator results in a set of planar graphs, since the procedure starts by planarizing the graph; it is likely that a tighter bound of  $O(\sqrt{g}||G||_\delta)$ would not result in a set of planar graphs.

\paragraph{Acknowledgements} We thank Anastasios Sidiropoulos and Mark Walsh for helpful discussions.  This material is based upon work supported by
the National Science Foundation under Grant Nos.\
CCF-0963921 and CCF-1054779.

\bibliographystyle{plain} 
\bibliography{ballcover}

\begin{thebibliography}{10}

\bibitem{AAPPS97}
P.~Agarwal, B.~Aronov, J.~Pach, R.~Pollack, and M.~Sharir.
\newblock Quasi-planar graphs have a linear number of edges.
\newblock {\em Combinatorica}, 17(1):1--9, 1997.

\bibitem{planarballcover}
Victor Chepoi, Bertrand Estellon, and Yann Vaxes.
\newblock Covering planar graphs with a fixed number of balls.
\newblock {\em Discrete \& Computational Geometry}, 37(2):237--244, 2007.

\bibitem{HMinorFree_JACM}
Erik~D. Demaine, Fedor~V. Fomin, MohammadTaghi Hajiaghayi, and Dimitrios~M.
  Thilikos.
\newblock Subexponential parameterized algorithms on graphs of bounded-genus
  and $h$-minor-free graphs.
\newblock {\em Journal of the ACM}, 52(6):866--893, 2005.

\bibitem{DDSV93}
K.~Diks, H.N. Djidjev, O.~S\'ykora, and I.~Vrt'o.
\newblock Edge separators of planar and outerplanar graphs with applications.
\newblock {\em Journal of Algorithms}, 14(2):258 -- 279, 1993.

\bibitem{Edmonds60}
J.~Edmonds.
\newblock A combinatorial representation for polyhedral surfaces.
\newblock {\em Notices of the American Mathematical Society}, 7:646, 1960.

\bibitem{Felsner04}
Stefan Felsner.
\newblock Topological graphs: Crossing lemma and applications.
\newblock In {\em Geometric Graphs and Arrangements}, Advanced Lectures in
  Mathematics, pages 43--52. Vieweg+Teubner Verlag, 2004.

\bibitem{FPS13}
Jacob Fox, J{\'a}nos Pach, and Andrew Suk.
\newblock The number of edges in k-quasi-planar graphs.
\newblock {\em SIAM J. Discrete Math.}, 27(1):550--561, 2013.

\bibitem{GPRS01}
Cyril Gavoille, David Peleg, André Raspaud, and Eric Sopena.
\newblock Small k-dominating sets in planar graphs with applications.
\newblock In Andreas Brandstädt and VanBang Le, editors, {\em Graph-Theoretic
  Concepts in Computer Science}, volume 2204 of {\em Lecture Notes in Computer
  Science}, pages 201--216. Springer Berlin Heidelberg, 2001.

\bibitem{GaMi90}
Hillel Gazit and Gary~L. Miller.
\newblock Planar separators and the {E}uclidean norm.
\newblock In {\em SIGAL International Symposium on Algorithms}, pages 338--347,
  Tokyo, August 1990. Information Processing Society of Japan, Springer-Verlag.

\bibitem{Guy1968376}
Richard~K. Guy, Tom Jenkyns, and Jonathan Schaer.
\newblock The toroidal crossing number of the complete graph.
\newblock {\em Journal of Combinatorial Theory}, 4(4):376 -- 390, 1968.

\bibitem{h-at-02}
Allen Hatcher.
\newblock {\em Algebraic Topology}.
\newblock Cambridge Univ. Press, 2002.

\bibitem{Helly23}
E.~Helly.
\newblock \"{U}ber mengen konvexer k\"{o}rper mit gemeinschaftlichen punkten.
\newblock {\em Jahresbericht der Deutschen Mathematiker-Vereinigung},
  32:175--176, 1923.

\bibitem{Matousek04}
Jirí Matou\u{s}ek.
\newblock Bounded {VC}-dimension implies a fractional {H}elly theorem.
\newblock {\em Discrete \& Computational Geometry}, 31(2):251--255, 2004.

\bibitem{Miller86}
G.~L. Miller.
\newblock Finding small simple cycle separators for 2-connected planar graphs.
\newblock {\em Journal of Compute and System Sciences}, 32(3):265--279, 1986.

\bibitem{mt-gs-01-ch4}
Bojan Mohar and Carsten Thomassen.
\newblock {\em Graphs on Surfaces}, chapter~4.
\newblock Johns Hopkins University Press, Baltimore, 2001.

\bibitem{m-t-00}
James~R. Munkres.
\newblock {\em Topology}.
\newblock Prentice-Hall, 2nd edition, 2000.

\bibitem{PSS94}
J\'{a}nos Pach, Farhad Shahrokhi, and Mario Szegedy.
\newblock Applications of the crossing number.
\newblock In {\em Proceedings of the tenth annual Symposium on Computational
  geometry}, SCG '94, pages 198--202, New York, NY, USA, 1994. ACM.

\bibitem{ringel-youngs68}
G.~Ringel and J.~W.~T. Youngs.
\newblock Solution of the {H}eawood map-coloring problem.
\newblock {\em Proceedings of the National Academy of Sciences of the United
  States of America}, 60:438--445, 1968.

\bibitem{Robertson200343}
Neil Robertson and P.D Seymour.
\newblock Graph minors. {XVI}. excluding a non-planar graph.
\newblock {\em Journal of Combinatorial Theory, Series B}, 89(1):43 -- 76,
  2003.

\bibitem{SSSV96}
F.~Shahrokhi, O.~S\'ykora, L.A. Sz\'ekely, and I.~Vrt'o.
\newblock The crossing number of a graph on a compact 2-manifold.
\newblock {\em Advances in Mathematics}, 123(2):105 -- 119, 1996.

\bibitem{SSSV94}
F.~Shahrokhi, L.~Sz\'ekely, O.~S\'ykora, and I.~Vrt'o.
\newblock Improved bounds for the crossing numbers on surfaces of genus g.
\newblock In Jan Leeuwen, editor, {\em Graph-Theoretic Concepts in Computer
  Science}, volume 790 of {\em Lecture Notes in Computer Science}, pages
  388--395. Springer Berlin Heidelberg, 1994.

\bibitem{SW13}
A.~Suk and B.~Walczak.
\newblock New bounds on the maximum number of edges in k-quasi-planar graphs.
\newblock In {\em Proc. of the Int'l Symp. on Graph Drawing}, 2013.

\bibitem{talk}
S.~Thomass\'{e}.
\newblock Invited talk, cana{DAM}, {U}niversity of {V}ictoria, May 2011.

\bibitem{Turan77}
P.~Tur\'{a}n.
\newblock A note of welcome.
\newblock {\em J. Graph Theory}, 1:7--9, 1977.

\bibitem{VC71}
V.~Vapnik and A.~Chervonenkis.
\newblock On the uniform convergence of relative frequencies of events to their
  probabilities.
\newblock {\em Theory of Probability and its Applications}, 16(2):264--280,
  1971.

\bibitem{Youngs63}
J.~Youngs.
\newblock Minimal imbeddings and the genus of a graph.
\newblock {\em Journal of Mathematical Mechanic}, 12:303--315, 1963.

\end{thebibliography}

\end{document}